%% file: dctl.tex
\documentclass[conference]{IEEEtran}
\usepackage{dctl}
\usepackage{comment}
\usepackage{balance}
\usepackage{dblfloatfix}

\begin{document}
\title{A Team Based Variant of CTL}

\author{
\IEEEauthorblockN{Andreas Krebs}
\IEEEauthorblockA{Wilhelm-Schickard-Institut für Informatik\\
Eberhard-Karls-Universität Tübingen\\
Sand 13, 72076 Tübingen, Germany\\
Email: \url|mail@krebs-net.de|}
%
\and
\IEEEauthorblockN{Arne Meier}
\IEEEauthorblockA{Institut für Theoretische Informatik\\
Leibniz Universität Hannover\\
Appelstr.~4, 30176 Hannover, Germany\\
\url|meier@thi.uni-hannover.de|}
\and
\IEEEauthorblockN{Jonni Virtema}
\IEEEauthorblockA{Institut für Theoretische Informatik\\
Leibniz Universität Hannover\\
\vspace{1mm}
Appelstr.~4, 30176 Hannover, Germany\\
School of Information Sciences\\
University of Tampere\\
Kanslerinrinne 1, 33014 Tampere, Finland\\
Email: \url|jonni.virtema@gmail.com|}
}

\maketitle

\begin{abstract}
 We introduce two variants of computation tree logic CTL based on team semantics: an asynchronous one and a synchronous one. For both variants we investigate the computational complexity of the satisfiability as well as the model checking problem. The satisfiability problem is shown to be EXPTIME-complete. Here it does not matter which of the two semantics are considered. For model checking we prove a PSPACE-completeness for the synchronous case, and show P-completeness for the asynchronous case. Furthermore we prove several interesting fundamental properties of both semantics.
\end{abstract}

\section{Introduction}

Temporal logic can be traced back to the late 1950s when Prior considered more formally the interplay of time and modality \cite{pr57}. Today it is a well-known and important logic in the area of computer science that has influenced the area of program verification significantly. Since the introduction of temporal logic a wide research field around temporal logic has emerged. The most seminal contributions in this field have been made by Kripke \cite{kri63}, Pnueli \cite{pn77}, Emerson, Clarke, and Halpern \cite{emha85,clem81} to name a few.

In real life applications, especially in the field of program verification, computational complexity is of the greatest significance. In the framework of logic, the most important related decision problems are the satisfiability problem and the model checking problem.
From a software engineering point of view the satisfiability problem can be seen as the question of specification consistency: The specification of a program is expressed via a formula of some logic (e.g., computation tree logic $\CTL$). One then asks whether there exists a model that satisfies the given formula. For model checking an implementation of a system is depicted via a Kripke structure and a specification via a formula of some logic. One then wants to know whether the structure satisfies the formula (i.e., whether the system satisfies the specification). The satisfiability problem for $\CTL$ is known to be $\EXPTIME$-complete by \citeauthor{fila79}, and   \citeauthor{pr80} \cite{fila79,pr80} whereas the model checking problem has been shown to be $\P$-complete by \citeauthor{clemsi86}, and \citeauthor{sc02} \cite{clemsi86,sc02}.

%
%
    Team semantics was introduced to the framework of first-order logic by Hodges \cite{Hodges97c} in the late 1990s. Subsequently V\"a\"an\"anen adopted the notion of a team as a core notion, first, for his (first-order) \emph{dependence logic} \cite{va07} and later, in the framework of modal logic, for \emph{modal dependence logic} \cite{va08}. The fundamental idea behind team semantics is crisp. The idea is to shift from \emph{singletons} to \emph{sets} as satisfying elements of formulas. These sets of satisfying elements are called \emph{teams}. In the \emph{team semantics} of first-order logic formulas are evaluated with respect to first-order structures and sets of assignments. In the \emph{team semantics} of modal logic formulas are evaluated with respect to Kripke structures and sets of worlds.
    
    Various logics with team semantics have been defined and investigated. Most of these logics are extensions of first-order, propositional, or modal logics with novel atomic propositions that describe properties of teams (e.g, inclusion, dependence, and independence). Modal dependence logic ($\MDL$) extends modal logic with \emph{propositional dependence atoms}. A dependence atom, denoted by $\dep{p_1,\dots,p_n,q}$, intuitively states that (inside a team) the truth value of the proposition $q$ is functionally determined by the truth values of the propositions $p_1,\dots,p_{n}$. It was soon realized that $\MDL$ lacks the ability to express temporal dependencies; there is no mechanism in $\MDL$ to express dependencies that occur between different points of the model. This is due to the restriction that only proposition symbols are allowed in the dependence atoms of modal dependence logic. To overcome this defect Ebbing et~al. \cite{ehmmvv13} introduced the
\emph{extended modal dependence logic} ($\EMDL$) by extending the scope of dependence atoms to arbitrary modal formulas. Dependence atoms of $\EMDL$ are of the form $\dep{\varphi_1,\dots\varphi_n,\psi}$, where $\varphi_1,\dots,\varphi_n,\psi$ are formulas of modal logic.

In recent years the research around first-order and modal team semantics has been vibrant. See, e.g., \cite{ehmmvv13,helusavi14,kmsv14} for related research in the modal context.
While team semantics has been considered in the context of regular modal logic, to the best knowledge of the authors, this is the first article to consider team semantics for a more serious temporal logic. The only logic from this area which can express some temporal like properties is $\EMDL$.

In this article we propose two team based variants of $\CTL$: an asynchronous one and a synchronous one. We abandon the idea of defining semantics for $\CTL$ via pointed Kripke structures. Instead the semantics are defined via pairs $(K,T)$, where $K$ is an ordinary Kripke structure and $T$, called \emph{a team of $K$}, is a subset of the domain of $K$. We will then investigate these two natural variants of $\CTL$ lifted to team semantics.

In the synchronous model we stipulate that the evolution of time is synchronous among all team members whereas in the asynchronous case we do not have this assumption.
The main difference of these two approaches can be seen in the definitions of the semantics for the modal operator \emph{until} (see \Cref{semantics}):
Either the time is synchronous among all team members, and hence when we quantify over a time point in the future all team members will advance the same number of steps in the Kripke structure, or we consider an asynchronous model, where when we quantify over a future point each team member might advance a different number of steps. We then investigate the expressive powers and computational complexity of these formalisms.



It remains to be seen whether the team-based semantics can be used to model computational
phenomena arising in the context of parallel or distributed processes.
Our logic should be viewed as a first adaptation of $\CTL$ in the context of team semantics. The next natural step is, of course, to add different dependency notions such as dependence and independence to the language. Describing dependency properties of computations is of great interest directly motivated from the area of dependence logic.  

\emph{Related work.}
There exists an approach of multi-modal $\CTL$, and one called alternating-time temporal logic $\mathcal{AT\!L}$. The first is a variant of $\CTL$ with several agents acting asynchronously. The latter is an extension of $\CTL$ that is used to reason about several agents acting synchronously (general concurrent game structures) or asynchronously (turn-based structures). For the first see, e.g., the work of \citeauthor{ahrsw09} \cite{ahrsw09}. The second contribution is due to the work of \citeauthor{aho02} \cite{aho02}.

Moreover a classification of the computational complexity of fragments of the satisfiability as well as the model checking problem of $\CTL$ by means of allowed Boolean operators and/or combinations of allowed temporal operators has been obtained recently \cite{mmtv09,bmmstv11}.
A survey on Kripke semantics with connections to several areas of logic, e.g., temporal, dependence, and hybrid logic can be found in a work of \citeauthor{mmmv12} \cite{mmmv12}.
An automatic-theoretic approach to branching-time model checking has been investigated by \citeauthor{kvw00} \cite{kvw00}. For a temporal logic with team-style semantics see the work of \citeauthor{aj07} \cite{aj07}.

In the literature, a multitude of approaches for modeling different kind of computation (e.g., serial, parallel, and distributed) have been considered. Also many natural connections to logic have been discovered.
Some of these approaches deal directly with computational devices as in circuit complexity (for details see,  e.g., \cite{vol99}). 
Another approach of this kind is the introduction of a parallel random access machine (PRAM) by \citeauthor{im89} \cite{im89}. Logical characterisations of complexity classes are investigated in the field of descriptive complexity theory. A multitude of natural characterisations are known (see, e.g., the book of \citeauthor{immer} \cite{immer} for further details).
A connection between particular modal logics and distributed computing has been considered recently by \citeauthor{disc15} \cite{disc15}. They give a characterisation of constant time parallel computation in the spirit of descriptive complexity.

\emph{Results.} We introduce two new variants of $\CTL$ based on team semantics: an asynchronous one and a synchronous one. We investigate the computational complexity of the satisfiability and the model checking problem of these variants. For model checking the complexity differs with respect to these variants. In the asynchronous case we show that the complexity is $\P$-complete and hence the same as for $\CTL$ by exploiting structural properties of the satisfaction relation. For synchronous semantics surprisingly the complexity becomes $\PSPACE$-complete. Hence having synchronised semantics makes the model checking in this logic intractable under reasonable complexity separation assumptions. For the satisfiability problem we show that the complexity stays $\EXPTIME$-complete (same as for $\CTL$) independently on which semantics is used.

\emph{Structure of the paper.} In \Cref{sec:prelims} we give syntax and semantics of two novel variants of computation tree logic $\CTL$. 
In \Cref{sec:properties} we prove closure properties of the satisfaction relations of the two variants. \Cref{sec:expressive} deals with their expressive power. In \Cref{sec:complexity} we completely classify the computational complexity of the satisfiability and the model checking problem with respect to both variants. Finally we present interesting further research directions and conclude.

\section{Preliminaries}\label{sec:prelims}
We start this section with a brief summary of the relevant complexity classes for this paper. We then define the syntax and semantics of computation tree logic $\CTL$. We deviate from the existing literature by using a convention that is customary related to logics with team semantics: We define the syntax of $\CTL$ in negation normal form, i.e., we require that negations may appear only in front of proposition symbols. We then introduce two variants of $\CTL$ that are designed to model parallel computation. \bigskip

\subsection{Complexity}
The underlying computation model is Turing machines. We will make use of the complexity classes $\P$, $\PSPACE$, and $\EXPTIME$. All reductions in this paper are \emph{logspace many-to-one reductions}, i.e., computable by a deterministic Turing machine running in logarithmic space. For a deeper introduction into this topic we refer the reader to the good book of \citeauthor{pip97} \cite{pip97}.

\subsection{Temporal Logic} 

Let $\PROP$ be a finite set of proposition symbols. The set of all $\CTL$-formulas is defined inductively via the following grammar:
$$
\varphi::= p\mid \lnot p\mid (\varphi\land\varphi)\mid(\varphi\lor\varphi)\mid\logicOpFont{P}\X\varphi\mid \logicOpFont{P}[\varphi\U\varphi]\mid \logicOpFont{P}[\varphi\W\varphi],
$$
where $\logicOpFont{P}\in\{\A,\E\}$ and $p\in\PROP$. We define the following usual shorthands: $\top:= p\vee \neg p$, $\bot:= p\wedge \neg p$, $\F \varphi := [\top\U\varphi]$, and $\G \varphi := [\varphi\W\bot]$. Note that the formulas are in negation normal form (NNF). This is not a severe restriction as transforming a given formula into its NNF requires linear time in the input length.

A Kripke structure $K$ is a tuple $(W,R,\eta)$ where $W$ is a finite, non-empty set of states, $R\colon W\times W$ is a total transition relation (i.e., for every $w\in W$ there is a $w'\in W$ such that $wRw'$), and $\eta\colon W\to2^{\PROP}$ is a labelling function. A path $\pi=\pi(1),\pi(2),\dots$ is an infinite sequence of states $\pi(i)\in W$ such that $\pi(i)R\pi(i+1)$ holds. By $\Pi(w)$ we denote the (possibly infinite) set of all paths $\pi$ for which $\pi(1)=w$.

\begin{definition}[Semantics of $\CTL$]
Let $K=(W,R,\eta)$ be a Kripke structure and $w\in W$ a state. The satisfaction relation $\models$ for $\CTL$ is defined as follows:
$$
\begin{array}{l@{}l}
 K,w\models p &\text{ iff } p\in\eta(w),\\
 K,w\models \lnot p &\text{ iff } p\notin\eta(w),\\
 K,w\models \varphi\land\psi &\text{ iff } K,w\models\varphi \text{ and }K,w\models\psi,\\
 K,w\models \varphi\lor\psi &\text{ iff } K,w\models\varphi \text{ or }K,w\models\psi,\\
 K,w\models\logicOpFont{P}\X\varphi &\text{ iff } \Game\pi\in\Pi(w):K,\pi(2)\models\psi,\\
 K,w\models\logicOpFont{P}[\varphi\U\psi] &\text{ iff } \Game\pi\in\Pi(w)\exists k\in\N:K,\pi(k)\models\psi\text{ and }\\
 &\qquad \forall 1\leq i<k: K,\pi(i)\models\varphi, \text{ and}\\
 K,w\models\logicOpFont{P}[\varphi\W\psi] &\text{ iff } \Game\pi\in\Pi(w)\forall i: K,\pi(i)\models\varphi\text{ or }\\
 &\qquad (\exists k\in\N:K,\pi(k)\models\psi\text{ and }\\
 &\qquad\;\forall 1\leq i<k: K,\pi(i)\models\varphi),
\end{array}$$
where $\logicOpFont{P}\in\{\A,\E\}$ and $\Game=\exists$ if $\logicOpFont{P}=\E$ and $\Game=\forall$ if $\logicOpFont{P}=\A$.
\end{definition}

Next we will introduce team semantics for $\CTL$ based on multisets. A \emph{multiset} is a generalisation of the concept of a set that allows multiple instances of the same element in the multiset. We denote a multiset that has elements $p$, $q$, $r$, and $r$ by $\mset{p,q,r,r}$. When $W$ is a set (or a multiset), we use $T\multisubseteq W$ to denote that $T$ is a multiset such that each element of $T$ is also an element of $W$. If $T,T'$ are multisets then $T\sqcup T'$ denotes the multiset defined by the disjoint union of the two multisets $T,T'$.

\begin{definition}[Team]
Let $K=(W,R,\eta)$ be a Kripke structure. Any multiset $T$ such that $T\multisubseteq W$ is called a \emph{team of $K$}.
\end{definition}

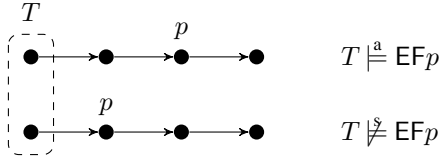
\begin{figure}
 \centering
\begin{tikzpicture}[c/.style={circle,fill=black,inner sep=0mm,minimum width=2mm},x=1cm,y=1cm]

\foreach \x/\l in {1/,2/,3/$p$,4/}
 \node[c,label={90:\l}] (a\x) at (\x,1) {};

\foreach \x/\l in {1/,2/$p$,3/,4/}
 \node[c,label={90:\l}] (b\x) at (\x,0) {};

\foreach \f/\t in {1/2,2/3,3/4}{
 \path[-stealth',black] (a\f) edge (a\t);
 \path[-stealth',black] (b\f) edge (b\t);
}

\draw[black,dashed,rounded corners] (.7,1.3) rectangle (1.3,-.3);
\node at (1,1.6) {$T$};
\node[anchor=west] at (5,1) {$T\amodels \EF p$};
\node[anchor=west] at (5,0) {$T\nsmodels\EF  p$};
\end{tikzpicture}
\caption{Difference between asynchronous and synchronous semantics shown with respect to the formula $\EF p$.}\label{fig:async-vs-synch}
\end{figure}

Next we define two semantics for $\CTL$ based on team semantics: an asynchronous one and a synchronous one. The difference can be seen in the clauses for until and weak until and is also depicted in \Cref{fig:async-vs-synch}.

\begin{definition}[Synchronous and asynchronous team semantics]\label{semantics}
Let $K=(W,R,\eta)$ be a Kripke structure, $T=\mset{t_1,\dots,t_n}$ be a team of $K$, and $\varphi$ and $\psi$ be $\CTL$-formulas. The synchronous satisfaction relation $\smodels$ and the asynchronous satisfaction relation $\amodels$ for $\CTL$ are defined as follows.
The following clauses are common to both semantics. In the clauses $\vdash$ denotes either $\smodels$ or $\amodels$.
$$
\begin{array}{lcl}
 K,T\vdash p & \text{ iff } & \forall w\in T: p\in\eta(w).\\
 K,T\vdash \lnot p & \text{ iff }& \forall w\in T: p\notin\eta(w).\\
 K,T\vdash \varphi\land\psi & \text{ iff } & K,T\vdash\varphi \text{ and }K,T\vdash\psi.\\
 K,T\vdash \varphi\lor\psi & \text{ iff } & \exists T_1\sqcup T_2=T\text{ such that } \\
 && K,T_1\vdash\varphi \text{ and }K,T_2\vdash\psi.\\
 K,T\vdash\E\X\varphi & \text{ iff } & \exists\pi_{t_1}\in\Pi(t_1),\dots,\exists\pi_{t_n}\in\Pi(t_n)\\
 && K,\multibigcup_{1\leq j\leq n}\mset{\pi_{t_j}(2)}\vdash\psi.\\
 K,T\vdash\A\X\varphi & \text{ iff } & \forall\pi_{t_1}\in\Pi(t_1),\dots,\forall\pi_{t_n}\in\Pi(t_n)\\
 && K,\multibigcup_{1\leq j\leq n}\mset{\pi_{t_j}(2)}\vdash\psi.
 \end{array}
$$ 

For the synchronous semantics we have the following clauses, where $\logicOpFont{P}\in\{\A,\E\}$, and  $\Game=\forall$ if $\logicOpFont{P}=\A$, resp., $\Game=\exists$ if $\logicOpFont{P}=\E$.
 \begin{align*}
 &K,T\smodels\logicOpFont{P}[\varphi\U\psi] \quad\text{ iff }\\
 &\quad \Game\pi_{t_1}\in\Pi(t_1),\dots,\Game\pi_{t_n}\in\Pi(t_n):\\
 &\qquad \exists k\in\N:K,\multibigcup_{1\leq j\leq n}\mset{\pi_{t_j}(k)}\smodels\psi\text{ and }\\
 &\qquad \forall 1\leq i<k: K,\multibigcup_{1\leq j\leq n}\mset{\pi_{t_j}(i)}\smodels\varphi.\\
  &K,T\smodels\logicOpFont{P}[\varphi\W\psi] \quad\text{ iff }\\
  &\quad \Game\pi_{t_1}\in\Pi(t_1),\dots,\Game\pi_{t_n}\in\Pi(t_n):\\
 &\qquad \forall i: K,\multibigcup_{1\leq j\leq n}\mset{\pi_{t_j}(i)}\smodels\varphi\text{ or }\\
 &\qquad (\exists k\in\N:K,\multibigcup_{1\leq j\leq n}\mset{\pi_{t_j}(k)}\smodels\psi\text{ and }\\
 &\qquad \forall 1\leq i<k: K,\multibigcup_{1\leq j\leq n}\mset{\pi_{t_j}(i)}\smodels\varphi).
 \end{align*}
 
For the asynchronous semantics we have the following clauses, where $\logicOpFont{P}\in\{\A,\E\}$, and  $\Game=\forall$ if $\logicOpFont{P}=\A$, resp., $\Game=\exists$ if $\logicOpFont{P}=\E$.
\begin{align*}
  &K,T\amodels\logicOpFont{P}[\varphi\U\psi] \quad\text{ iff }\\ 
 &\quad \Game\pi_{1}\in\Pi(t_1),\dots,\Game\pi_{n}\in\Pi(t_n),\; \exists k_{1},\dots,\exists k_{n}\in\N:\\
 &\qquad K,\multibigcup_{1\leq j\leq n}\mset{\pi_{j}(k_{j})}\amodels\psi\text{ and }\\
 &\qquad \forall 1\leq i_1<k_1,\dots,\forall1\leq i_n<k_n: K,\!\!\!\!\multibigcup_{1\leq j\leq n}\!\!\!\!\mset{\pi_{j}(i_j)}\amodels\varphi.\\
 &K,T\amodels\logicOpFont{P}[\varphi\W\psi] \quad\text{ iff }\\
 &\quad \Game\pi_{1}\in\Pi(t_1),\dots,\Game\pi_{n}\in\Pi(t_n):\\
 &\qquad \forall i_1,\dots,\forall i_n: K,\multibigcup_{1\leq j\leq n}\{\pi_{j}(i_j)\}\amodels\varphi\text{ or }\\
 &\qquad (\exists k_{1},\dots,\exists k_{n}\in\N: K,\multibigcup_{1\leq j\leq n}\mset{\pi_{j}(k_{j})}\amodels\psi\text{ and }\\
 &\quad\qquad \forall 1\leq i_1\!<\!k_1,\dots,\forall 1\!\leq\! i_n\!<\!k_n\!:\!K,\!\!\!\!\multibigcup_{1\leq j\leq n}\!\!\!\!\mset{\pi_{j}(i_j)}\amodels\varphi).
\end{align*}

\end{definition}

Observe that the Boolean connective $\lor$ removes synchronicity between the team members.

\section{Properties of Asynchronous and Synchronous Semantics}\label{sec:properties}
In the following section we investigate several properties of the asynchronous and synchronous satisfaction relations $\amodels,\smodels$. In particular, we will use them in the end to deduce a corollary for asynchronous semantics which shows the interplay with the usual $\CTL$ satisfaction relation.\bigskip

Observe that $K, T \vdash \bot$ holds if and only if $T=\emptyset$. The proof of the following lemma then is very easy.

\begin{lemma}[Empty team property]
The following holds for every Kripke model $K$ and $\vdash$ in $\{\smodels, \amodels\}$:
\[
K, \emptyset \vdash \varphi \text{ holds for every $\CTL$-formula $\varphi$}.
\]
\end{lemma}

When restricted to singleton teams, the synchronous and asynchronous team semantics coincide with the traditional semantics of $\CTL$ defined via pointed Kripke models.\begin{lemma}[Singleton equivalence]\label{lem:singleton}
 For every Kripke structure $K=(W,R,\eta)$ and every world $w\in W$ the following equivalence holds:
 $$
 K,\mset{w}\amodels\varphi \overset{(1)}{\Leftrightarrow} K,\mset{w}\smodels\varphi\overset{(2)}{\Leftrightarrow} K,w\models\varphi.
 $$
\end{lemma}
\begin{proof}
It is straightforward to check that on singleton teams the synchronous semantics of until and weak until coincide with that of the asynchronized semantics. Since none of the clauses in the two semantics makes the size of teams grow, the equivalence (1) follows.

Now turn to (2). Let $K=(W,R,\eta)$ be an arbitrary Kripke structure. We first prove the claim via induction on structure of $\varphi$:

Assume that $\varphi$ is a (negated) proposition symbol $p$. Now
\begin{align*}
K,w&\models \varphi\\
&\text{iff}\quad \text{$p$ is (not) in $\eta(w')$}\\
&\text{iff}\quad \text{for all $w'\in \mset{w}$ it holds that $p$ is (not) in $\eta(w')$}\\
&\text{iff}\quad K,\mset{w}\smodels\varphi.
\end{align*}

The case $\land$ trivial. For the $\lor$ case, assume that $\varphi = \psi\lor\theta$. Now it holds that
\begin{align*}
K,&w\models \psi\lor\theta\\ 
&\text{iff}\quad K,w\models \psi \text{ or } K,w\models \theta\\
&\text{iff}\quad K,\mset{w}\smodels \psi \text{ or } K,\mset{w}\smodels \theta\\
&\text{iff}\quad (K,\mset{w}\smodels \psi \text{ and } K,\emptyset \smodels \theta) \text{ or } \\
& \hphantom{\text{iff}\quad} (K,\emptyset \smodels \psi \text{ and } K,\mset{w}\smodels \theta)\\
&\text{iff}\quad \exists T_1\sqcup T_2=\mset{w}\text{ s.t.\ }K,T_1\smodels\psi \text{ and }K,T_2\smodels\theta\\
&\text{iff}\quad K,\mset{w}\smodels \psi\lor\theta.
\end{align*}   

Here the first equivalence holds by the semantics of disjunction, the second equivalence follow by the induction hypothesis, the third via the empty set property, the fourth via the empty set property in combination with the semantics of ``or'', and the last by the team semantics of disjunction.

The cases for $\E\X$ and $\A\X$, until and weak until are all similar and straightforward. We show here the case for $\E\X$. Assume $\varphi = \E\X\psi$. Now $K,w\models \E\X\psi$ iff there exists a point $\pi\in\Pi(w)$ such that $K,\pi(2) \models \psi$. Now since trivially $\multibigcup_{1\leq j\leq 1}\mset{\pi_{t_j}(2)} = \mset{\pi_{t_1}(2)}$, and since by the induction hypothesis $K,\pi(2) \models \psi$ iff $K,\mset{\pi(2)} \smodels \psi$, the above is equivalent to $K,\mset{w}\smodels \E\X\psi$.
\end{proof}

Let $\vdash$ denote a team satisfaction relation. We say that $\vdash$ is \emph{downward closed} if the following holds for every Kripke structure $K$, for every $\CTL$-formula $\varphi$, and for every team $T$ and $T'$ of $K$:
\[
\text{If $K,T\vdash \varphi$ and $T'\multisubseteq T$ then $K,T'\vdash \varphi$}. 
\]

The proof of the following lemma is analogous with the corresponding proofs for modal and first-order dependence logic (see \cite{va07,va08}).
\begin{lemma}[Downward closure]\label{lem:dwclos}
 $\amodels$ and $\smodels$ are downward closed.
\end{lemma}
\begin{proof}
We proof the claim for $\smodels$ only. For $\amodels$ the argumentation is similar. The proof is by induction on the structure of $\varphi$.

Let $K=(W,R,\eta)$ be an arbitrary Kripke structure and $T'\subseteq T$ be some teams of $K$.
The cases for literals are trivial: Assume $K,T\vdash p$. Then by definition $p\in\eta(w)$ for every $w\in T$. Now since $T'\subseteq T$, clearly $p\in\eta(w)$ for every $w\in T'$. Thus $K,T'\vdash p$. The case for negated propositions symbols is completely symmetric.
 
The case for $\land$ is clear. For the case for $\varphi\lor\psi$ assume that $K,T\smodels\varphi\lor\psi$. Now by the definition of disjunction there exist $T_1\cup T_2=T$ such that $K,T_1\smodels\varphi$ and $K,T_2\smodels\psi$. By induction hypothesis it the follows that $K,T_1\cap T'\smodels\varphi$ and $K,T_2\cap T'\smodels\psi$. Now since clearly $T'= (T_1\cap T') \cup (T_2\cap T')$, it follows by the semantics of the disjunction that $K,T'\smodels\varphi\lor\psi$.

Now consider $\mathsf P\X\varphi$. Let $T=\mset{t_1,\dots,t_n}$, where $n\in\N$, and assume that $K,T\smodels\mathsf\P\X\varphi$. We have to show that $K,T'\smodels\mathsf P\X\varphi$ for every $T'\subseteq T$. By the semantics of $\mathsf P\X\varphi$ we have that
\begin{equation}\label{eq:1}
\Game\pi_{t_1}\in\Pi(t_1),\dots, \pi_{t_n}\in\Pi(t_n): K,\hspace{-1mm}\multibigcup_{1\leq j\leq n}\hspace{-1mm}\mset{\pi_{t_j}(2)}\smodels\varphi.
\end{equation}

It suffices to show that for every subteam $T'=\mset{s_1,\dots,s_k}$ of $T$, with $1\leq k\leq n$,
\[
\Game\pi_{s_1}\in\Pi(s_1),\dots,\pi_{s_n}\in\Pi(s_k): K,\multibigcup_{1\leq j\leq k}\mset{\pi_{t_j}(2)}\smodels\varphi
\]
holds.
But this follows from \eqref{eq:1} by the induction hypothesis. The cases for $\U$ and $\W$ are analogous.
\end{proof}

In this article, we consider multisets of points as teams. Observe that with respect to the satisfaction relation the use of multisets has no real consequence. However this does not hold for all extensions of these variants (see, e.g., \cite{va14}). The proof of the following corollary is self-evident. The proof uses the fact that both satisfaction relations are downward closed.
\begin{corollary}
Let $\varphi$ be a $\CTL$-formula, $\vdash\in\{\smodels,\amodels\}$, $K$ be a Kripke structure, $T$ be a team of $K$, and $T'$ be the underlying set of the multiset $T$. Then $K,T\vdash \varphi$ iff $K,T'\vdash \varphi$.
\end{corollary}

A team satisfaction relation $\vdash$ is said to be \emph{union closed} if for every Kripke structure $K$, formula $\varphi$, and teams $T$ and $T'$ of $K$, the following holds:
\[
\text{If $K,T\vdash\varphi$ and $K,T'\vdash\varphi$ then $K,T\sqcup T'\vdash\varphi$.}
\]

\begin{lemma}[Union closure]
 $\amodels$ is union closed.
\end{lemma}
\begin{proof}
 This property again can be shown via induction on structure on $\varphi$. The interesting parts of the proof are the cases for the temporal operators $\mathsf P[\varphi\U\psi]$ and $\mathsf P[\varphi\W\psi]$. We will show the proof for $\mathsf P[\varphi\U\psi]$ only. The proof for $\W$ is completely analogous. Now let $K,T\amodels\mathsf P[\varphi\U\psi]$ and $K,T'\amodels\mathsf P[\varphi\U\psi]$. For simplicity we show the result only for $\mathsf P=\E$. Let $T=\mset{t_1,\dots,t_n}$ be a team. Then $K,T\amodels\mathsf \E[\varphi\U\psi]$ implies that there are paths $\pi_1\in\Pi(t_1),\dots,\pi_n\in\Pi(t_n)$ and natural numbers $k_1,\dots,k_n$ such that $K,\mset{\pi_j(k_j)}\amodels\psi$ and for all $1\leq i_j<k_j$ it holds that $K,\mset{\pi_j(i_j)}\amodels\varphi$ for $1\leq j\leq n$. Analogously let $T'=\mset{s_1,\dots,s_m}$ be a team. Then $K,T'\amodels\E[\varphi\U\psi]$ implies that there are paths $\pi'_1\in\Pi(s_1),\dots,\pi'_m\in\Pi(s_m)$ and natural numbers $k'_1,\dots,k'_m$ such that $K,\mset{\pi'_j(k_j)}\amodels\psi$ and for all $1\leq i'_j<k'_j$ it holds that $K,\mset{\pi'_j(i_j)}\amodels\varphi$ for $1\leq j\leq m$. Thus clearly $K,T\sqcup T'\amodels\E[\varphi\U\psi]$ and the claim follows.
\end{proof}
Note that the semantics $\smodels$ is not union closed due to the observation depicted in \Cref{fig:async-vs-synch}. 

The previous lemmas lead to the following interesting corollary which allows one to consider only the elements of the team instead of the complete team together. 
This will later prove to be important in the classification of the complexity of the model checking problem for asynchronous semantics.

\begin{corollary}\label{cor:asynch-single}
 For every Kripke structure $K=(W,R,\eta)$ and every team $T$ of $K$ the following equivalence holds:
 $$
 K,T\amodels\varphi\Leftrightarrow \forall t\in T: K,t\models\varphi.
 $$
\end{corollary}

\section{Expressive power}\label{sec:expressive}
In this section, we discuss in more details the relationship between the expressive powers of team $\CTL$ with the synchronous semantics and team $\CTL$ with the asynchronous semantics.

\begin{definition}
For each $\CTL$-formula $\varphi$, define
\begin{align*}
&\mathfrak{F}^a_\varphi:=\{(K,T)\mid K,T\amodels\varphi\} \text{ and}\\
&\mathfrak{F}^s_\varphi:=\{(K,T)\mid K,T\smodels\varphi\}.
\end{align*}

We say that $\varphi$ \emph{defines the class $\mathfrak{F}^a_\varphi$ in asynchronous semantics} (of $\CTL$). Analogously, we say that $\varphi$ \emph{defines the class $\mathfrak{F}^s_\varphi$ in synchronous semantics} (of $\CTL$). A class $\mathfrak{F}$ of pairs of Kripke structures and teams is \emph{definable in asynchronous semantics} (in synchronous semantics), if there exists some $\psi\in\CTL$ such that $\mathfrak{F}=\mathfrak{F}^a_\psi$ ($\mathfrak{F}=\mathfrak{F}^s_\psi$). Furthermore, for $k\in\mathbb{N}$, define
\begin{align*}
\mathfrak{F}^{a,k}_\varphi&:=\{(K,T)\mid K,T\amodels\varphi \text{ and } \lvert T\rvert \leq k \}, \text{ and}\\
\mathfrak{F}^{s,k}_\varphi&:=\{(K,T)\mid K,T\smodels\varphi \text{ and } \lvert T\rvert \leq k \}.
\end{align*}

We say that $\varphi$ \emph{$k$-defines the class $\mathfrak{F}^{a,k}_\varphi$ (resp., $\mathfrak{F}^{s,k}_\varphi$) in asynchronous (resp., synchronous) semantics} (of $\CTL$). The definition of \emph{$k$-definability} is analogous to that of definability. 
\end{definition}
Next we will show that there exists a class $\mathfrak{F}$ which is definable in asynchronous semantics, but is not definable in synchronous semantics.
\begin{theorem}
The class $\mathfrak F_{\EF p}^a$ is not definable in synchronous semantics.
\end{theorem}
\begin{proof}
For the sake of a contradiction, assume that $\varphi$ is such that $\mathfrak F^a_\varphi = \mathfrak F_{\EF p}^s$. Consider the following Kripke model $K=(W,R,V)$, where $W=\{1,2,3,4\}$, $R=\{(1,4),(4,4),(2,3),(3,3)\}$, and $V(p)=\{1,3\}$. Clearly $K,\mset{1}\smodels\EF p$ and $K,\mset{2}\smodels\EF p$. Thus by our assumption, it follows that $K,\mset{1}\amodels\varphi$ and $K,\mset{2}\amodels\varphi$. From Corollary \ref{cor:asynch-single} it then follows that $K,\mset{1,2}\amodels\varphi$. But clearly $K,\mset{1,2}\nsmodels\EF p$.
\end{proof}
\begin{corollary}
For $k>1$, the class $\mathfrak F_{\EF p}^{a,k}$ is not $k$-definable in synchronous semantics.
\end{corollary}

\begin{conjecture}\label{conj:asynch}
The class $\mathfrak F_{\EF p}^s$ is not definable in asynchronous semantics.
\end{conjecture}

\begin{theorem}
For every $k\in\mathbb{N}$ and $\varphi\in\CTL$, the class $\mathfrak F_{\varphi}^{s,k}$ is $k$-definable in asynchronous semantics.
\end{theorem}
\begin{proof}
Fix $k\in \mathbb{N}$ and $\varphi\in\CTL$. Define
\[
\varphi' := \bigvee_{1\leq i \leq k} \varphi.
\]

We will show that $\mathfrak F_{\varphi}^{a,k} = \mathfrak F_{\varphi'}^{s,k}$. Let $K$ be an arbitrary Kripke structure and $T$ be a team of $K$ of size at most $k$. Then it holds
\begin{align*}
K,T \amodels \varphi \quad\Leftrightarrow\quad& \forall w\in T : K,\mset{w} \amodels \varphi \\
\Leftrightarrow\quad& \forall w\in T : K,\mset{w} \smodels \varphi \\
\Leftrightarrow\quad& K,T \smodels \varphi'.
\end{align*}

The first equivalence follows by \Cref{cor:asynch-single}, the second by \Cref{lem:singleton}, and the last by the semantics of disjunction and the downward closure property.
\end{proof}

\section{Complexity Results}\label{sec:complexity}
In this section we classify the problems with respect to the computational complexity. At first we start with the asynchronous semantics. We will begin with model checking and will finish with satisfiability.

In the following we define the most important decision problems in these logics.
\problemdef{$\aMC$}{A Kripke structure $K$, a team $T$ of $K$, a formula $\varphi\in\CTL$.}{$K,T\amodels\varphi$?}


\problemdef{$\aSAT$}{A formula $\varphi\in\CTL$.}{Does there exists a Kripke structure $K$ and a non-empty team $T$ of $K$ s.t.\ $K,T\amodels\varphi$?}


%

Similarly we write $\sMC$, resp., $\sSAT$ for the variants with synchronized semantics.

\subsection{Model Checking}
In this subsection we investigate the computational complexity of model checking. For usual $\CTL$ model checking the following proposition summarizes what is known.

\begin{proposition}[\cite{clemsi86,sc02}]\label{prop:MC-P}
 Model checking for $\CTL$ formulas is $\P$-complete.
\end{proposition}

At first we investigate the case for asynchronous semantics. Through combinations of the previous structural properties of $\amodels$ it is possible to show the same complexity degree.

\begin{theorem}
 $\aMC$ is $\P$-complete.
\end{theorem}
\begin{proof}
 The lower bound is immediate from usual CTL model checking by \Cref{prop:MC-P}. For the upper bound we apply \Cref{cor:asynch-single} and separately use for each member of the given team the usual CTL model checking algorithm.
\end{proof}

Now we turn to the model checking problem for synchronous semantics. Here we show that the problem becomes intractable under reasonable complexity class separation assumptions, i.e., $\P\neq\PSPACE$. The main idea is to exploit the synchronous semantics in a way to literally check in parallel all clauses for a given quantified Boolean formula for satisfiability for a set of relevant assignments.

\newcommand{\lits}{\textrm{literals}}
\newcommand{\var}{\textrm{var}}
\begin{theorem}\label{lem:mcs-pspace}
 $\sMC$ is $\PSPACE$-hard.
\end{theorem}
\begin{proof}
From \citeauthor{st77} \cite{st77} we know that the validity problem of closed quantified Boolean formulas ($\QBFVAL$) of the form $\exists x_{1}\forall x_{2}\cdots\Game x_{n}F$, where $\Game=\exists$ if $n$ is odd, resp., $\Game=\forall$ if $n$ is even, and $F$ is in conjunctive normal form is $\PSPACE$-complete.

 Let $\varphi\dfn\exists x_{1}\forall x_{2}\cdots\Game x_{n}\bigwedge_{i=1}^{m}\bigvee_{j=1}^{3}\ell_{i,j}$ be a closed quantified Boolean formula ($\QBF$) and $\Game=\exists$ if $n$ is odd, resp., $\Game=\forall$ if $n$ is even. Now define the corresponding structure $(W,R,\eta)$ as follows (also see \Cref{fig:generalview-pspace-struc}):
\begin{align*}
 W &\dfn \bigcup_{i=1}^{n}(\{w^{x_{i}}_{j}\mid 1\leq j\leq i\}\cup\{w^{x_{i}}_{j,1},w^{x_{i}}_{j,2}\mid i< j\leq n+4\})\\
 &\hspace{1.7em} \cup\{w^{c}_{i}\mid 1\leq i\leq n+1\}\cup\{w^{c_j}\mid 1\leq j\leq m\}\\
 &\hspace{1.7em} \cup\{w^{c_j}_{j,i,k}\mid 1\leq j\leq m,1\leq i\leq 3,1\leq k\leq 2\},\displaybreak\\
 R &\dfn \bigcup_{i=1}^{n}(\{(w^{x_{i}}_{j},w^{x_{i}}_{j+1})\mid 1\leq j<i\}\\
 &\quad\qquad\cup\{(w^{x_{i}}_{i},w^{x_{i}}_{i+1,1}),(w^{x_{i}}_{i},w^{x_{i}}_{i+1,2})\}\\
 &\quad\qquad \cup \{(w^{x_{i}}_{j,k},w^{x_{i}}_{j+1,k})\mid 1\leq k\leq 2, i<j\leq n+3\})\\
 &\quad\qquad \cup \{(w^{x_{i}}_{n+4,k},w^{x_{i}}_{n+4,k})\mid 1\leq k\leq 2\})\\
 &\hspace{1.7em} \cup \{(w^{c}_{i},w^{c}_{i+1})\mid 1\leq i < n\}\\
 &\hspace{1.7em} \cup\{(w^c_{n+1},w^{c_j})\mid 1\leq j\leq m\}\\
 &\hspace{1.7em} \cup \{(w^{c_j},w^{c_j}_{j,i,1}),(w^{c_j}_{j,i,1},w^{c_j}_{j,i,2})\mid 1\leq\! i\!\leq 3,1\leq\!j\!\leq m\}\\
 &\hspace{1.7em} \cup \{(w^{c_j}_{j,i,2},w^{c_j}_{j,i,2})\mid 1\leq\! i\!\leq 3,1\leq\!j\!\leq m\},\text{ and}\\
 \eta&\dfn \big\{(w^{x_i}_{n+3,1},\{x_i\}\cup\{x_k\mid 1\leq k\neq i\leq n\})\;\big|\; 1\leq i \leq n\big\}\\
 &\hspace{1.7em} \cup \big\{(w^{x_i}_{n+4,2},\{x_i\}\cup\{x_k\mid 1\!\leq\! k\!\neq\! i\!\leq\! n\}))\;\big|\; 1\leq\! i\! \leq\! n\big\}\\
 &\hspace{1.7em} \cup \big\{(w^{c_j}_{j,i,1}, \{x_k\mid \ell_{j,i}=x_k\}\\
 &\qquad\quad\cup\{x_k\mid x_k\neq\var(\ell_{j,i})\})\;\big|\; 1\leq j\leq m,1\leq i\leq 3\big\}\\
 &\hspace{1.7em} \cup \big\{(w^{c_j}_{j,i,2}, \{x_k\mid \ell_{j,i}=\lnot x_k\}\\
 &\qquad\quad\cup\{x_k\mid x_k\neq\var(\ell_{j,i})\})\;\big|\; 1\leq j\leq m,1\leq i\leq 3\big\}.
\end{align*}

Further set 
\begin{align*}
  T&\dfn\mset{w^{x_{1}}_{1},\dots,w^{x_{n}}_{1},w_{1}^{c}} \text{ and }\\
 \varphi&\dfn\underbrace{\EX\AX\cdots\mathsf{P}\X}_n\AX\EX\bigwedge_{i=1}^{n}\EF x_{i},
\end{align*}
where $\mathsf{P}=\E$ if $n$ is odd and $\mathsf{P}=\A$ if $n$ is even. Let the reduction be defined as $f\colon\langle\varphi\rangle\mapsto\langle (W,R,\eta),T,\varphi\rangle$. 

In \Cref{fig:example-pspace} an example of the reduction is shown for the instance $\exists x_{1}\forall x_{2}\exists x_{3}(x_{1}\lor\overline{x_{2}}\lor\overline{x_{3}})\land(\overline{x_{1}}\lor x_{2}\lor x_{3})\land(\overline{x_{1}}\lor\overline{x_{2}}\lor\overline{x_{3}})$. Note that this formula is a valid $\QBF$ and hence belongs to $\QBFVAL$. The left three branching systems choose the values of the $x_{i}$s. A decision for the left/right path can be interpreted as setting variable $x_{i}$ to 1/0.

\begin{figure*}[!b]
\centering
\input{general-structure-pspace-hardness}

\caption{General view on the created Kripke structure in the proof of \Cref{lem:mcs-pspace}.}
\label{fig:generalview-pspace-struc}
\end{figure*}
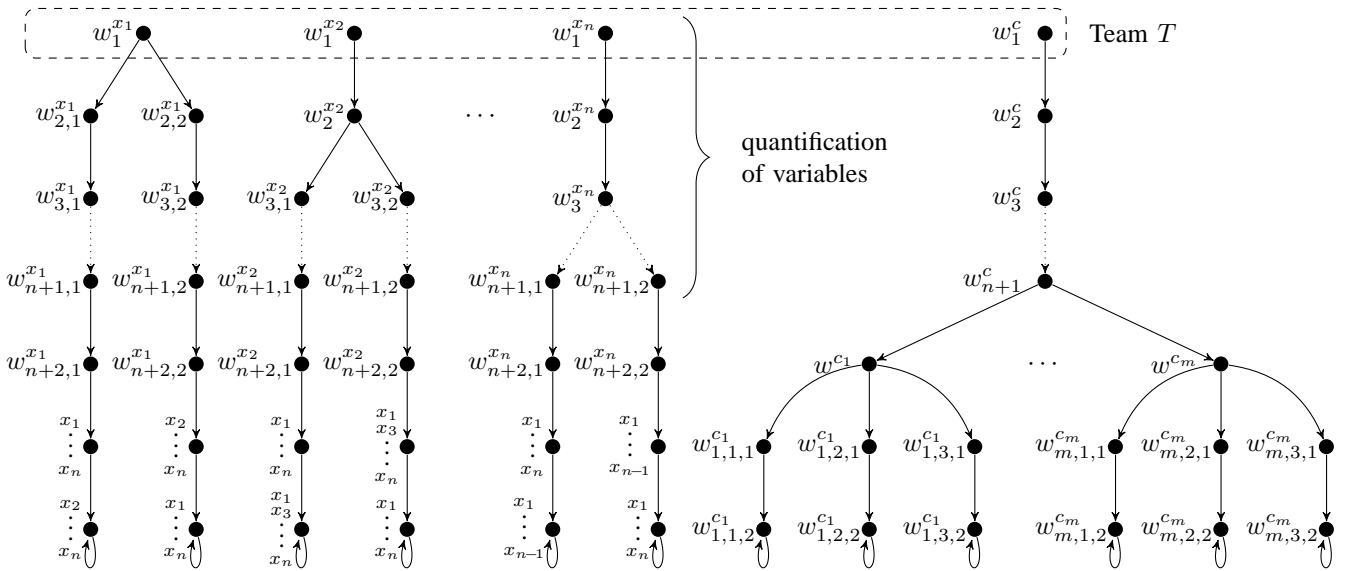

For the correctness of the reduction we need to show that $\varphi\in\QBFVAL$ iff $f(\varphi)\in\sMC$.

``$\Rightarrow$'': Let $\varphi\in\QBFVAL$, $\varphi=\exists x_{1}\forall x_{2}\cdots\Game x_{n}F$, $F=\bigwedge_{i=1}^{m}\bigvee_{j=1}^{3}\ell_{i,j}$, and let $S$ be a valid set of assignments with respect to $\exists x_{1}\forall x_{2}\cdots\Game x_{n}$. Now it holds that for every $s\in S$ that $s\models F$ holds. Choose an arbitrary such $s\in S$. Note that the variables now can be seen as being existentially quantified with respect to every assignment in $S$ (whereas strictly speaking some of them stem from a universal quantifier $\forall$, yet at the moment we consider only a single assignment). Denote with $f(\varphi)=\langle(W,R,\eta),T,\varphi\rangle$ the value of the reduction function and denote with $K$ the structure $(W,R,\eta)$. 

Now we will prove that $K,T\smodels\varphi$. Observe that $T=\mset{w_{1}^{x_{1}},\dots,w_{1}^{x_{n}},w_{1}^{c}}$ by definition. For $w_{1}^{c}$ there is no choice in the next $n$ steps defined by the prefix of $\varphi$. For $w_{1}^{x_{1}},\dots,w_{1}^{x_{n}}$ we decide as follows depending on the assignment $s$. 

Note that during the evaluation of $\varphi$ w.r.t.\ $T$ and $K$ in the first $n$ $\CTL$ operators of $\varphi$ the $\AX$ operators are treated in the proof now as $\EX$. This is because here we just have to decide with respect to the chosen assignment from $S$. Hence if $s(x_{i})=1$ then choose in step $i$ of this prefix from $w^{x_{i}}_{i}$ the successor world $w^{x_{i}}_{i+1,1}$. If $s(x_{i})=0$ then choose $w^{x_{i}}_{i+1,2}$ instead. 

Now after $n$ steps the current team $T'$ then is $\mset{w^{c}_{n+1}}\multicup\mset{w^{x_{i}}_{n+1,1}\mid s(x_{i})=1,1\leq i\leq n}\multicup\mset{w^{x_{i}}_{n+1,2}\mid s(x_{i})=0,1\leq i\leq n}$ (note that now the team completely agrees with the assignment $s$). In the next step the team branches now on all clauses of $F$ and becomes $\mset{w^{c_{j}}\mid 1\leq j\leq m}\multicup\mset{w^{x_{i}}_{n+2,1}\mid s(x_{i})=1,1\leq i\leq n}\multicup\mset{w^{x_{i}}_{n+2,2}\mid s(x_{i})=0,1\leq i\leq n}$. Now continuing with an $\EX$ in $\varphi$ the team members of the ``formula'' (we here refer to the elements $\mset{w^{c_{j}}\mid 1\leq j\leq m}$ of the team) have to decide for a literal which satisfies the respective clause. As $s\models F$ this must be possible. W.l.o.g.\ assume that in clause $C_{j}$ the literal $\ell_{j}$ satisfies $C_{j}$ by $s(\ell_{j})=1$ for $1\leq j\leq m$ (denote with $s(\ell)$ the value $1-s(x)$ if $x$ is the corresponding variable to literal $\ell$). Let $\textrm{index}(\ell_{j})\in\{1,2,3\}$ denote the ``index'' of $\ell_{j}$ in $C_{j}$, i.e., the value $i\in\{1,2,3\}$ such that $\ell_{j}=\ell_{i,j}$ in $F$. Then we choose the world $w^{c_{j}}_{j,\textrm{index}(\ell_{j}),1}$ as a successor from $w^{c_{j}}$ for $1\leq j\leq m$. 

For the (``variable'' team members) $w^{x_{i}}_{n+2,k}$ with $k\in\{1,2\}$ we have no choice and proceed to $w^{x_{i}}_{n+3,k}$. Now we have to satisfy the remainder of $\varphi$ which is $\bigwedge_{i=1}^{n}\EF x_{i}$. Observe that for variable team members $w^{x_{i}}_{n+3,1}$ only has $x_{i}$ labeled in the current world and not in the successor world $w^{x_{i}}_{n+4,1}$, i.e., $x_{i}\notin \eta(w^{x_{i}}_{n+4,1})$. 

Symmetrically this is true for the $w^{x_{i}}_{n+3,2}$ worlds but $x_{i}\notin \eta(w^{x_{i}}_{n+3,2})$ and $x_{i}\in \eta(w^{x_{i}}_{n+4,2})$. Hence ``staying'' in the world (hence immediately satisfying the $\EF x_{i}$) means setting $x_{i}$ to true by $s$ whereas making a further step means setting $x_{i}$ to false by $s$. 

Further observe for the formula team members we have depending on the value of $s(\ell_{j})$ that $x\in \eta(w^{c_{j}}_{n+3,\textrm{index}(\ell_{j}),1})$ and $x\notin \eta(w^{c_{j}}_{n+3,\textrm{index}(\ell_{j}),2})$ if $s(\ell_{j})=1$, and $x\notin \eta(w^{c_{j}}_{n+3,\textrm{index}(\ell_{j}),1})$ and $x\in \eta(w^{c_{j}}_{n+3,\textrm{index}(\ell_{j}),2})$ if $s(\ell_{j})=0$. Thus according to synchronous semantics the step depth w.r.t.\ a $x_{i}$ have to be the same for every element of the team. Hence if we decided for the variable team member that $s(x_{i})=1$ then for the formula team members we cannot make a step to the successor world and therefore have to stay (similarly if $s(x_{i})=0$ then we have to do this step). 

Note that this is not relevant for other states as there all variables are labelled as propositions and are trivially satisfied everywhere. Hence as $\ell_{j}\models C_{j}$ we have decided for the world $w^{x_{i}}_{n+3,2-\textrm{index}(\ell_{j})}$ and can do a step if $s(\ell_{j})=0$ and stay if $s(\ell_{j})=1$. Hence $K,T\smodels\varphi$.

For the direction ``$\Leftarrow$'' observe that with similar arguments we can deduce from the ``final'' team in the end what has to be a satisfying assignment depending on the choices of $w^{x_{i}}_{n+3,k}$ and $k\in\{1,2\}$. Hence by construction any of these assignments satisfies $F$. Let again denote by $S$ a set of teams which satisfy $\AX\EX\bigwedge_{i=1}^{n} x_{i}$ according to the prefix of $n$ $\CTL$ operators. Then define a set $S'$ of assignments from $S$ by getting the assignment $s$ from the team $t\in S$ by setting $s(x_{i})=1$ if there is a world $w^{x_{i}}_{n+1,1}$ in $t$ and otherwise $s(x_{i})=0$. Then it analogously follows that $s\models F$. $S'$ also agrees on the quantifier prefix of $\varphi$. Hence $\varphi\in\QBFVAL$.
\end{proof}

\begin{figure}\centering
\begin{tikzpicture}[c/.style={circle,fill=black,inner sep=0mm,minimum width=1.5mm},x=.48cm,y=.8cm]

 \foreach \x/\y/\n/\l in {0/0/0/,-.5/-1/l1/,.5/-1/r1/,-.5/-2/l2/,.5/-2/r2/,-.5/-3/l3/,.5/-3/r3/,-.5/-4/l4/,.5/-4/r4/,%
 	-.5/-5/l5/$\substack{x_{1}\\x_{2}\\x_{3}}$,.5/-5/r5/$\substack{x_{2}\\x_{3}}$,%
	-.5/-6/l6/$\substack{x_{2}\\x_{3}}$,.5/-6/r6/$\substack{x_{1}\\x_{2}\\x_{3}}$}
  \node[c,label={[xshift=1.3mm]180:\footnotesize\l}] (a\n) at (\x,\y) {};

 \foreach \x/\y/\n/\l in {2.5/0/0/,2.5/-1/l1/,2/-2/l2/,3/-2/r2/,2/-3/l3/,3/-3/r3/,2/-4/l4/,3/-4/r4/,%
 	2/-5/l5/$\substack{x_{1}\\x_{2}\\x_{3}}$,3/-5/r5/$\substack{x_{1}\\x_{3}}$,%
	2/-6/l6/$\substack{x_{1}\\x_{3}}$,3/-6/r6/$\substack{x_{1}\\x_{2}\\x_{3}}$}
  \node[c,label={[xshift=1.3mm]180:\footnotesize\l}] (b\n) at (\x,\y) {};

 \foreach \x/\y/\n/\l in {5/0/0/,5/-1/l1/,5/-2/l2/,4.5/-3/l3/,5.5/-3/r3/,4.5/-4/l4/,5.5/-4/r4/,%
 	4.5/-5/l5/$\substack{x_{1}\\x_{2}\\x_{3}}$,5.5/-5/r5/$\substack{x_{1}\\x_{2}}$,%
	4.5/-6/l6/$\substack{x_{1}\\x_{2}}$,5.5/-6/r6/$\substack{x_{1}\\x_{2}\\x_{3}}$}
  \node[c,label={[xshift=1.3mm]180:\footnotesize\l}] (c\n) at (\x,\y) {};

 \foreach \f/\t in {0/l1,0/r1,l1/l2,l2/l3,l3/l4,l4/l5,l5/l6,r1/r2,r2/r3,r3/r4,r4/r5,r5/r6}{
  \path[-stealth',black] (a\f) edge (a\t);
 }
 \foreach \f/\t in {0/l1,l1/l2,l1/r2,l2/l3,l3/l4,l4/l5,l5/l6,r2/r3,r3/r4,r4/r5,r5/r6}{
  \path[-stealth',black] (b\f) edge (b\t);
 }
 \foreach \f/\t in {0/l1,l1/l2,l2/l3,l2/r3,l3/l4,l4/l5,l5/l6,r3/r4,r4/r5,r5/r6}{
  \path[-stealth',black] (c\f) edge (c\t);
 }

 \foreach \x/\y/\n/\l in 
 {12/0/0/,
 12/-1/1/,
 12/-2/2/,
 12/-3/3/,
 8.5/-4/l4/,
 12/-4/m4/,
 15.5/-4/r4/,%
 7.5/-5/ll5/$\substack{x_{1}\\x_{2}\\x_{3}}$,
 8.5/-5/lm5/$\substack{x_{1}\\x_{3}}$,
 9.5/-5/lr5/$\substack{x_{1}\\x_{2}}$,
 7.5/-6/ll6/$\substack{x_{2}\\x_{3}}$,
 8.5/-6/lm6/$\substack{x_{1}\\x_{2}\\x_{3}}$,
 9.5/-6/lr6/$\substack{x_{1}\\x_{2}\\x_{3}}$,
 11/-5/ml5/$\substack{x_{2}\\x_{3}}$,
 12/-5/mm5/$\substack{x_{1}\\x_{2}\\x_{3}}$,
 13/-5/mr5/$\substack{x_{1}\\x_{2}\\x_{3}}$,
 11/-6/ml6/$\substack{x_{1}\\x_{2}\\x_{3}}$,
 12/-6/mm6/$\substack{x_{1}\\x_{3}}$,
 13/-6/mr6/$\substack{x_{1}\\x_{2}}$,
 14.5/-5/rl5/$\substack{x_{2}\\x_{3}}$,
 15.5/-5/rm5/$\substack{x_{1}\\x_{3}}$,
 16.5/-5/rr5/$\substack{x_{1}\\x_{2}}$,
 14.5/-6/rl6/$\substack{x_{1}\\x_{2}\\x_{3}}$,
 15.5/-6/rm6/$\substack{x_{1}\\x_{2}\\x_{3}}$,
 16.5/-6/rr6/$\substack{x_{1}\\x_{2}\\x_{3}}$}
  \node[c,label={[xshift=1.3mm]180:\footnotesize\l}] (d\n) at (\x,\y) {};

 \foreach \f/\t in {0/1,1/2,2/3,3/l4,3/m4,3/r4,%
 	l4/lm5,ll5/ll6,lm5/lm6,lr5/lr6,%
	m4/mm5,ml5/ml6,mm5/mm6,mr5/mr6,%
	r4/rm5,rl5/rl6,rm5/rm6,rr5/rr6}{
  \path[-stealth',black] (d\f) edge (d\t);
 }
 
 \foreach \x in {al6,ar6,bl6,br6,cl6,cr6,dll6,dlm6,dlr6,dml6,dmm6,dmr6,drl6,drm6,drr6}{ 
  \path[-stealth',black] (\x) edge[>=stealth',loop below,] (\x);
 }
 
 \foreach \f/\t in {l4/ll5,m4/ml5,r4/rl5}{
  \path[-stealth',black] (d\f) edge[bend right] (d\t);
 }
 
 \foreach \f/\t in {l4/lr5,m4/mr5,r4/rr5}{
  \path[-stealth',black] (d\f) edge[bend left] (d\t);
 }
 
 \draw[black,dashed,rounded corners] (-.3,-.3) rectangle (12.3,.3);
 \node[anchor=west] at (12.5,0) {Team $T$};
 
 \node[text width=2cm,align=center] at (8,-3) {\footnotesize agreed assignment};
 \draw[black,dashed] (-1,-3) -- (7.5,-3);
 
\end{tikzpicture}
 \caption{Example structure built in proof of Lemma~\ref{lem:mcs-pspace}.}\label{fig:example-pspace}
\end{figure}
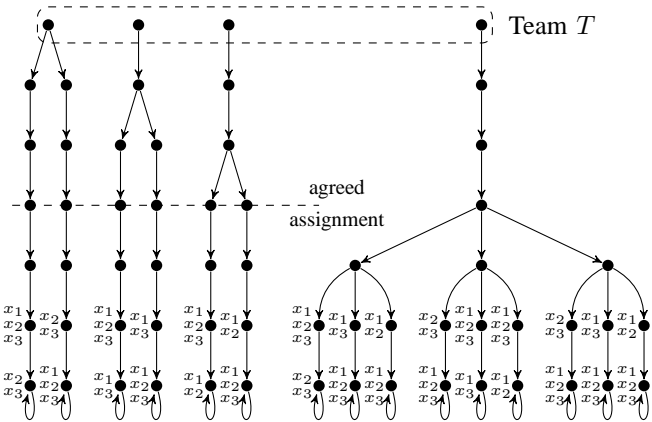

\begin{theorem}
 $\sMC$ is in $\PSPACE$.
\end{theorem}
\begin{proof}
 The following $\PSPACE$-algorithm solves $\sMC$. The weak until cases are omitted as they can be defined analogously to the usual until cases and just use non-determinism to operate on the disjunction.
 
\SetKw{proce}{Procedure}
\SetKwFunction{smc}{s-check}
\SetKwFunction{succ}{succ}

\LinesNumbered
\begin{algorithm}\caption{Model checking algorithm for $\sMC$}\label{alg:TeamMC}
\proce \succ{Structure $K=(W,R,\eta)$, team $T$}\;
 guess a multiset $T'$ with $|T|=|T'|$ s.t.~f.a.~$t\in T$ there exists a $u\in T'$ with $tRu$ and vice versa\;
 \Return $T'$\;
 \BlankLine
\proce \smc{Kripke structure $K=(W,R,\eta)$, team $T$, formula $\varphi$}\;
 \lIf{$\varphi=\true$}{\Return 1}
 \lIf{$\varphi=\false$}{\Return $T=\emptyset$?}
 \lIf{$\varphi=p$}{\Return $\forall w \in T:p\in \eta(w)$?}
 \lIf{$\varphi=\lnot p$}{\Return $\forall w \in T:p\notin \eta(w)$?}
 \If{$\varphi=\alpha\land\beta$}{\Return \smc{$K,T,\alpha$}$\land$\smc{$K,T,\beta$}}
 \If{$\varphi=\alpha\lor\beta$}{
  guess $T_1\sqcup T_2=T$\;
  \Return \smc{$K,T_1,\alpha$}$\land$\smc{$K,T_2,\beta$}\;
 }
 \If{$\varphi=\mathsf{EX}\alpha$}{
  $T'\leftarrow$\succ{$K,T$}\;
  \Return \smc{$K,T',\alpha$}\;
 }
 \If{$\varphi=\mathsf{AX}\alpha$}{
  bool $v\leftarrow1$\;
  \For{every possible guess $T'\leftarrow$\succ{$K,T$}}{
   $v\leftarrow v\land$\smc{$K,T',\alpha$}\;
  }
  \Return $v$\;
 }
 \If{$\varphi=\mathsf{E}[\alpha\U\beta]$}{
  guess a binary number $k\in [0,|W|^{|T|}]$, $v\leftarrow1$, and $T_{\text{last}}\leftarrow T$\;
  \For{$1\leq i \leq k$}{
   $T'\leftarrow$\succ{$K,T_{{\text{last}}}$}\;  
   $v\leftarrow v\land$\smc{$K,T',\alpha$}\;
   $T_{\text{last}}\leftarrow T'$ and $i\leftarrow i+1$\;
   \lIf{$i=k$}{$T_{\text{last}}\leftarrow$\succ{$K,T_{{\text{last}}}$}}
  }
  \Return $v\land$\smc{$K,T_{\text{last}},\beta$}\;
 }
 \If{$\varphi=\mathsf{A}[\alpha\U\beta]$}{
  let $T_{\text{last}}\leftarrow T$, and bool $v\leftarrow1$\;
  \For{$1\leq i \leq |W|^{|T|}$}{
   \For{every possible guess $T'\leftarrow$\succ{$K,T$}}{
    $T_{\text{last}}\leftarrow T'$, and guess $A\in\{0,1\}$\;
    \leIf{A}{$v\leftarrow v\land$\smc{$K,T',\alpha$}\;}{\Return $v\land$\smc{$K,T',\beta$}}
   }
   $i\leftarrow i+1$\;
  }
  \Return $v\land$\smc{$K,T',\beta$}\;
 }

\end{algorithm}

The procedure \texttt{s-check} (see Algorithm \ref{alg:TeamMC}) computes for a given Kripke structure $K$, a team $T$ and a formula $\varphi$ if $K,T\smodels\varphi$.

The correctness of the algorithm can be verified by induction over the formula $\varphi$
as the different cases in the procedure \texttt{s-check} merely restate the semantical definition of our team logic.

For the case $\varphi=\mathsf{E}[\alpha\U\beta]$ by definition we need to check if there exists paths
$\pi_{t_1}\in\Pi(t_1),\dots,\pi_{t_n}\in\Pi(t_n)$
and a $k\in\N$ such that
\begin{align*}
& K,\multibigcup_{1\leq j\leq n}\mset{\pi_{t_j}(k)}\smodels\beta\text{ and }\\
& \forall 1\leq i<k: K,\multibigcup_{1\leq j\leq n}\mset{\pi_{t_j}(i)}\smodels\alpha.
\end{align*}
The algorithm checks exactly the same conditions, but guesses the number $k$ only up to $|W|^{|T|}$. We show this is sufficient as the size $|T|$ of the team does not increase in the process of evaluation. Suppose such a $k$ exists but $k>|W|^{|T|}$, then there are $i_1<i_2$ such that all paths have a loop from $i_1$ to $i_2$, i.e.,
$$\forall{1\leq j\leq n}: \pi_{t_j}(i_1)=\pi_{t_j}(i_2).$$ 
We can generate a new set of paths
$\pi'_{t_1}\in\Pi(t_1),\dots,\pi'_{t_n}\in\Pi(t_n)$ by removing the loop from $i_1$ to $i_2$
and let $k'=k-i_2+i_1$. Then these paths and the new constant $k'$ also satisfy the conditions above.
We can repeat this process until we gained a constant less then $|W|^{|T|}$. Hence if there is such a $k$ we can find a $k\leq|W|^{|T|}$.

Similar it suffices in the case $\varphi=\mathsf{A}[\alpha\U\beta]$ to verify that $\beta$ is satisfied after at most $|W|^{|T|}$ steps.

Also our algorithm runs in alternating polynomial time; the nondeterministic choices occur in the Until-case and in the procedure \texttt{succ}, where they correspond to existential and universal quantifications. Hence the algorithm runs in $\PSPACE$.
\end{proof}

\begin{corollary}
 $\sMC$ is $\PSPACE$-complete.
\end{corollary}

\subsection{Satisfiability}
The following proposition summarises what is known about usual $\CTL$ satisfiability.

\begin{proposition}[\cite{fila79,pr80}]\label{prop:CTLSAT-complexity}
 Satisfiability for $\CTL$ formulas is $\EXPTIME$-complete.
\end{proposition}

For the team based variants of computation tree logic the computational complexity of the satisfiability problem is proven to be the same as for $\CTL$.

\begin{theorem}\label{thm:SAT}
 $\sSAT$ and $\aSAT$ are $\EXPTIME$-complete.
\end{theorem}
\begin{proof}
 In both cases the problem merely asks whether there exists a Kripke structure $K$ and a non-empty team $T$ of $K$ such that $K,T\amodels\varphi$, resp., $K,T,\smodels\varphi$ for given formula $\varphi\in\CTL$. By \Cref{lem:dwclos} we can just quantify for a singleton sized team, i.e., $|T|=1$. By \Cref{lem:singleton} we immediately obtain the same complexity bounds from usual satisfiability for $\CTL$. Hence \Cref{prop:CTLSAT-complexity} applies and proves the theorem.
\end{proof}

\section{Future Work}
The tautology or validity problem for this new logic is quite interesting and seems to have a higher complexity than the related satisfiability problem. However we have not been able to prove a result yet. Formally the corresponding problems are defined as follows:
\problemdef{$\aVAL$}{A formula $\varphi\in\CTL$.}{Does $K,T\amodels\varphi$ hold for every Kripke structure $K$ and every team $T$ of $K$?}

\problemdef{$\sVAL$}{A formula $\varphi\in\CTL$.}{Does $K,T\smodels\varphi$ hold for every Kripke structure $K$ and every team $T$ of $K$?}

In the context of team-based propositional and modal logics the computational complexity of the validity problem has been studied by \citeauthor{v14} \cite{v14}. \citeauthor{v14} shows that the problem for propositional dependence logic is $\NEXPTIME$-complete whereas for (extended) modal dependence logic it is $\NEXPTIME$-hard and in $\NEXPTIME^{\NP}$.
 
One might also consider to settle \Cref{conj:asynch} which we left open. Intuitively here the weak until operator makes the argument quite difficult to prove due to the possibility of infinite computation paths (informally hence the $\G$ operator).

Dependence logic (construed broadly) is a prospering area in logic in which team semantics has been extensively studied. The logic itself was introduced by \citeauthor{va07} \cite{va07} in 2007 with an aim to express dependencies between variables in systems. Subsequently multitude of related formalisms have been defined and studied. There are several fruitful applications areas for these formalisms, e.g., computational biology, database systems, social choice theory, and cryptography. A modal logic variant of dependence logic was defined by \citeauthor{va08} \cite{va08} in 2008. Modal dependence logic is an extension of modal logic with novel atomic propositions called dependence atoms. 
A dependence atom, denoted by $\dep{p_1,\dots,p_n,q}$, intuitively states that (inside a team) the truth value of the proposition $q$ is functionally determined by the truth values of the propositions $p_1,\dots,p_{n}$.
The introduction of dependence atoms (or some other dependency notions from the field of dependence logic) into our team-based $\CTL$ might lead to a flexible and elegant variant of computation tree logic which can express several interesting dependency properties relevant to practice. Formally dependence atoms are defined in our formalism as follows. If $K=(W,R,\eta)$ is a Kripke structure, $T=\mset{t_{1},\dots,t_{n}}$ is a team, and $\varphi_{1},\dots,\varphi_{n}$ are $\CTL$ formulas, then $K,T\models \textrm{dep}(\varphi_{1},\dots,\varphi_{n})$ holds if and only if
\begin{align*}
 &\forall t_{1},t_{2}\in T: \bigwedge_{i=1}^{n-1}(K,\mset{t_{1}}\models \varphi_{i}\Longleftrightarrow K,\mset{t_{2}}\models\varphi_{i})\\
 &\text{ implies } (K,\mset{t_{1}}\models\varphi_{n}\Longleftrightarrow K,\mset{t_{2}}\models\varphi_{n}).
\end{align*}
The above is the definition of what is known as \emph{modal dependence atoms} of extended modal dependence logic $\mathcal{EMDL}$ introduced by \citeauthor{ehmmvv13} \cite{ehmmvv13}.
 
 It is well-known that there are several alternative inputs to consider in the model checking problem. In general, a model and a formula are given, and then one needs to decide whether the model satisfies the formula. \emph{System complexity} considers the computational complexity for the case of a fixed formula whereas \emph{specification complexity} fixes the underlying Kripke structure. We considered in this paper the \emph{combined complexity} where both a formula and a model belong to the given input. Yet the other two approaches might give more specific insights into the intractability of the synchronous model checking case we investigated. In particular the study of so-to-speak \emph{team complexity}, where the team or the team size is assumed to be fixed, might as well be of independent interest.
 
 Finally this leads to the consideration of different kinds of restrictions on the problems. In particular for the quite strong $\PSPACE$-completeness result for model checking in synchronous semantics it is of interest where this intractability can be pinned to. Hence the investigation of fragments by means of allowed temporal operators and/or Boolean operators will lead to a better understanding of this presumably untamable high complexity.
 
\section{Conclusion}\label{sec:conclusion}
 In this paper we studied computation tree logic in the context of team semantics. We identified two alternative definitions for the semantics: an asynchronous one and a synchronous one. The intuitive difference between these semantics is that, in the latter semantics the flow of time can be seen as synchronous inside a team, whereas in the former semantics the flow of time inside a team can differ. This difference manifests itself in the semantical clauses for the eventuality operator \emph{until} as well as for the temporal operator \emph{future}. From satisfiability perspective the complexity of the new logics behave similar as $\CTL$. One might consider a different kind of satisfiability question: given a formula $\varphi$ and a team size $k$, does there exist a Kripke structure $K$ and a team $T$ of size $k$ in $K$ such that $K,T\smodels\varphi$, resp., $K,T\amodels\varphi$? However the use of the multiset notion easily tames this approach and then lets us conclude with the same result as in \Cref{thm:SAT}.
 
For model checking the complexity of the synchronous case differs to the one of usual $\CTL$. This fact stems from the expressive notion of synchronicity between team members and is in line with the results of \citeauthor{kvw00} \cite{kvw00}. We prove $\PSPACE$-completeness. The lower bound follows by a reduction from $\QBF$ validity and the upper bound via a Ladner-style algorithm \cite{lad77}. 
It might first seem that the complexity of the asynchronous case would also differ with the quite efficient $\CTL$ case (which is $\P$-complete). However the use of closure properties of the relation $\amodels$ enables us to separately check, for each team member, whether it is satisfied in the given structure. Thus a multiple application of the usual $\CTL$ model checking algorithm thereby establishes the same upper bound.
 
\section*{Acknowledgements}
The second author is supported by DFG grant ME 4279/1-1.
The third author supported by Jenny and Antti Wihuri Foundation.
We thank the anonymous referees for their helpful comments.
\bibliographystyle{plainnat}
\bibliography{dctl}
\end{document}

%% file: general-structure-pspace-hardness.tex
\usetikzlibrary{decorations.pathreplacing}

\begin{tikzpicture}[c/.style={circle,fill=black,inner sep=0mm,minimum width=2mm},x=.85cm,y=1cm,scale=1.1]

\begin{scope}[transform canvas={xshift = -1cm}]
 \foreach \x/\y/\n/\l in 
  {0.25/0/0/$w_{1}^{x_{1}}$,
   -.5/-1/l1/$w_{2,1}^{x_{1}}$,
   1/-1/r1/$w_{2,2}^{x_{1}}$,
   -.5/-2/l2/$w_{3,1}^{x_{1}}$,
   1/-2/r2/$w_{3,2}^{x_{1}}$,
   -.5/-3/l3/$w_{n+1,1}^{x_{1}}$,
   1/-3/r3/$w_{n+1,2}^{x_{1}}$,
   -.5/-4/l4/$w_{n+2,1}^{x_{1}}$,
   1/-4/r4/$w_{n+2,2}^{x_{1}}$,
   -.5/-5/l5/$\substack{x_{1}\\{ \vphantom{\int\limits^x}\smash\vdots}\\x_{n}}$,
   1/-5/r5/$\substack{x_{2}\\{ \vphantom{\int\limits^x}\smash\vdots}\\x_{n}}$,
   -.5/-6/l6/$\substack{x_{2}\\{ \vphantom{\int\limits^x}\smash\vdots}\\x_{n}}$,
   1/-6/r6/$\substack{x_{1}\\{ \vphantom{\int\limits^x}\smash\vdots}\\x_{n}}$}
  \node[c,label={[xshift=1.3mm]180:\l}] (a\n) at (\x,\y) {};

 \foreach \f/\t in {0/l1,0/r1,l1/l2,l3/l4,l4/l5,l5/l6,r1/r2,r3/r4,r4/r5,r5/r6}{
  \path[-stealth',black] (a\f) edge (a\t);
 }

 \path[-stealth',dotted,black] (ar2) edge (ar3);
 \path[-stealth',dotted,black] (al2) edge (al3);

 \foreach \x/\y/\n/\l in 
  {3.25/0/0/$w_{1}^{x_{2}}$,
   3.25/-1/l1/$w_{2}^{x_{2}}$,
   2.5/-2/l2/$w_{3,1}^{x_{2}}$,
   4/-2/r2/$w_{3,2}^{x_{2}}$,
   2.5/-3/l3/$w_{n+1,1}^{x_{2}}$,
   4/-3/r3/$w_{n+1,2}^{x_{2}}$,
   2.5/-4/l4/$w_{n+2,1}^{x_{2}}$,
   4/-4/r4/$w_{n+2,2}^{x_{2}}$,
   2.5/-5/l5/$\substack{x_{1}\\{ \vphantom{\int\limits^x}\smash\vdots}\\x_{n}}$,
   4/-5/r5/$\substack{x_{1}\\x_{3}\\{ \vphantom{\int\limits^x}\smash\vdots}\\x_{n}}$,
   2.5/-6/l6/$\substack{x_{1}\\x_{3}\\{ \vphantom{\int\limits^x}\smash\vdots}\\x_{n}}$,
   4/-6/r6/$\substack{x_{1}\\{ \vphantom{\int\limits^x}\smash\vdots}\\x_{n}}$}
  \node[c,label={[xshift=1.3mm]180:\l}] (b\n) at (\x,\y) {};

 \foreach \f/\t in {0/l1,l1/l2,l1/r2,l3/l4,l4/l5,l5/l6,r3/r4,r4/r5,r5/r6}{
  \path[-stealth',black] (b\f) edge (b\t);
 }

 \path[-stealth',dotted,black] (br2) edge (br3);
 \path[-stealth',dotted,black] (bl2) edge (bl3);
 
 \foreach \x in {al6,ar6,bl6,br6}{ 
  \path[-stealth',black] (\x) edge[>=stealth',loop below,] (\x);
 } 
\end{scope}

 \node at (4,-1) {$\cdots$};

 \foreach \x/\y/\n/\l in 
  {5.75/0/0/$w_{1}^{x_{n}}$,
  5.75/-1/l1/$w_{2}^{x_{n}}$,
  5.75/-2/l2/$w_{3}^{x_{n}}$,
  5/-3/l3/$w_{n+1,1}^{x_{n}}$,
  6.5/-3/r3/$w_{n+1,2}^{x_{n}}$,
  5/-4/l4/$w_{n+2,1}^{x_{n}}$,
  6.5/-4/r4/$w_{n+2,2}^{x_{n}}$,
  5/-5/l5/$\substack{x_{1}\\{ \vphantom{\int\limits^x}\smash\vdots}\\x_{n}}$,
  6.5/-5/r5/$\substack{x_{1}\\{ \vphantom{\int\limits^x}\smash\vdots}\\x_{n\!-\!1}}$,
  5/-6/l6/$\substack{x_{1}\\{ \vphantom{\int\limits^x}\smash\vdots}\\x_{n\!-\!1}}$,
  6.5/-6/r6/$\substack{x_{1}\\{ \vphantom{\int\limits^x}\smash\vdots}\\x_{n}}$}
  \node[c,label={[xshift=1.3mm]180:\l}] (c\n) at (\x,\y) {};

 \foreach \f/\t in {0/l1,l1/l2,l3/l4,l4/l5,l5/l6,r3/r4,r4/r5,r5/r6}{
  \path[-stealth',black] (c\f) edge (c\t);
 }

 \path[-stealth',dotted,black] (cl2) edge (cr3);
 \path[-stealth',dotted,black] (cl2) edge (cl3);

\draw [decorate,decoration={brace,amplitude=10pt,mirror,raise=4pt},yshift=0pt]
(6.7,-3.2) -- (6.7,0.2) node [black,midway,xshift=0.8cm,text width=2cm,anchor=west] {quantification\newline of variables};

 \foreach \x/\y/\n/\l in 
 {12/0/0/$w_{1}^{c}\;\;$,
 12/-1/1/$w_{2}^{c}\;\;$,
 12/-2/2/$w_{3}^{c}\;\;$,
 12/-3/3/$w_{n+1}^{c}\;\;$,
 9.5/-4/l4/$w^{c_{1}}\;$,
 14.5/-4/r4/$w^{c_{m}}\;\;$,
 8/-5/ll5/$w_{1,1,1}^{c_{1}}$,
 9.5/-5/lm5/$w_{1,2,1}^{c_{1}}$,
 11/-5/lr5/$w_{1,3,1}^{c_{1}}$,
 8/-6/ll6/$w_{1,1,2}^{c_{1}}$,
 9.5/-6/lm6/$w_{1,2,2}^{c_{1}}$,
 11/-6/lr6/$w_{1,3,2}^{c_{1}}$,
 13/-5/rl5/$w_{m,1,1}^{c_{m}}$,
 14.5/-5/rm5/$w_{m,2,1}^{c_{m}}$,
 16/-5/rr5/$w_{m,3,1}^{c_{m}}$,
 13/-6/rl6/$w_{m,1,2}^{c_{m}}$,
 14.5/-6/rm6/$w_{m,2,2}^{c_{m}}$,
 16/-6/rr6/$w_{m,3,2}^{c_{m}}$}
  \node[c,label={[xshift=1.3mm]180:\l}] (d\n) at (\x,\y) {};

 \foreach \f/\t in {0/1,1/2,3/l4,3/r4,
 	l4/lm5,ll5/ll6,lm5/lm6,lr5/lr6,%
	r4/rm5,rl5/rl6,rm5/rm6,rr5/rr6}{
  \path[-stealth',black] (d\f) edge (d\t);
 }
 
 \path[-stealth',dotted,black] (d2) edge (d3);

 \foreach \x in {cl6,cr6,dll6,dlm6,dlr6,drl6,drm6,drr6}{ 
  \path[-stealth',black] (\x) edge[>=stealth',loop below,] (\x);
 }
 
 \foreach \f/\t in {l4/ll5,r4/rl5}{
  \path[-stealth',black] (d\f) edge[bend right] (d\t);
 }
 
 \foreach \f/\t in {l4/lr5,r4/rr5}{
  \path[-stealth',black] (d\f) edge[bend left] (d\t);
 }

 \node at (12,-4) {$\cdots$};
 
 \draw[black,dashed,rounded corners] (-2.5,-.3) rectangle (12.3,.3);
 \node[anchor=west] at (12.5,0) {Team $T$};

\end{tikzpicture}